\documentclass[12pt,a4paper]{article}

\usepackage{cmbright,fullpage}
\usepackage{amsmath}%
\usepackage{amssymb}%
\usepackage{epsfig}%
\usepackage{amsthm}%
\usepackage{dsfont}%
\usepackage{natbib}
\usepackage{url}%
\usepackage{mathrsfs}%
\usepackage{braket}%
\usepackage{accents}
\usepackage{booktabs}
\usepackage{eurosym}
\usepackage{float}
\usepackage{placeins}

\usepackage{amssymb,latexsym,amsmath,soul,mathptmx,mathrsfs,verbatim}     
\usepackage[amsthm,mathsfit]{libertinust1math}
\newtheorem{rem}{\bf{Remark}}

\advance\voffset by 0pt
\advance\hoffset by 0pt
\theoremstyle{plain}%

\newtheorem{cor}{Corollary}

\newtheorem{proposition}{Proposition}

\theoremstyle{definition}%

\renewenvironment{appendix}%
{\setcounter{section}{0}%
\numberwithin{equation}{section}%
}%
{}%
\begin{document}

  \title{Status Quo Bias and Asset-Pricing under Model-Risk
  \thanks{The author thanks Prof. J. Thijssen, Dept. of Mathematics, University of York, for his patient advice.}
}
\author{KEN KANGDA WREN$^1$\\
$^1$Dept. of Mathematics, University of York}
\maketitle
\begin{abstract}
Existing works on asset-pricing under model-risk seek to adapt the utility-preference framework to generate prices that replicate the key anomalies of long-run historical data. We move the focus closer to practice by taking a bottom-up perspective and a Risk-Neutral Equivalent approach. 
Viable model-risk pricing is shown to have a '{canonical}' form, with an intuitive price-of-risk and dynamic. It allows the \emph{ex-post} effect of {bias} and risk-pricing to be disentangled without the usual joint-hypothesis problem
. As a result novel connections are made between a disparate group of subjects: Momentum, Low-Risk, Status Quo Bias and Bewley's Inertia Axiom under Knightian uncertainty. 

\end{abstract}
\newpage
\begin{center}
    {\huge Conflict-of-Interest Disclosure Statement}
\end{center}
\textbf{Author K.K.Wren}\\
I have nothing to disclose.
\newpage
\section*{Introduction}

\subsection*{Motivation}
We examine asset-pricing under model-risks that are 'uncertain' in the sense of \cite{Kni} and \cite{Bew}, where 'uncertainty' refers to randomness whose laws are {unknown}. For our purpose, 'risk' shall refer to randomness (perceived or real) under laws known or unknown, and 'model', some probability law to be chosen from a set of alternatives.

Making do with 'working models' refined over time is common in practice. Attempts to embed it into the standard theory of asset-pricing include \emph{Parameter Learning} (e.g. \cite{GTim}, \cite{CD}), \emph{Robust Control} (e.g. \cite{hs1}, \cite{hs2}), and \emph{Ambiguity Aversion} (e.g. \cite{ep95}, 
\cite{ep02}, \cite{miao}, \cite{hndbk}). Existing focus has been the nature of preference under model-risk and the generation of general equilibrium prices that reproduce the high equity-premium  (among other anomalies) seen in historical data. While fruitful in many ways, they have yet to impact mainstream theory and practice, arguably, due to complexity and a lack of connection to day-to-day investment tasks, which, ironically, must cope with model-risks constantly. 

We aim to move the subject matter closer to practice by taking a bottom-up perspective and a 'minimal' approach via the Risk-Neutral Equivalent (RNE) formulation, whose sole premise is the Fundamental Theorem of Asset Pricing (FTAP). RNE pricing rarely features in model-risk studies, although all the competing theories are by now FTAP-compliant (e.g. 
\cite{knoarb}) despite early difficulties (\cite{ep95}). The reasons are technical, rather than fundamental. First, \emph{viability}, No Arbitrage in the FTAP sense, is too weak for most purposes: there are far more {viable processes} than {economically valid} ones and extra conditions are almost always needed to 'pin down' the price dynamic. Second, learning can make RNE probabilities path-dependent and so intractable (e.g. \cite{GTim}). We show that these issues diminish in common settings that incorporate certain practical characteristics.

\subsection*{Scope, Premise and Results}
We take for granted that asset-prices are the result of utility maximisation by 'representative agents', 
who adhere to the laws of probability, but not necessarily objective accuracy; so the reference beliefs behind pricing may differ from their objective counterparts.

Consider binary model-risk \(B\in\{b,\overline{b}\}\), with probability law \(\mathbf{p}^B_0\), which may be unknown. Let our assets, when \(B\)-sure, have \emph{classical risks} and \emph{classical prices}. We shall focus on 'one-off' \(B\)-risks linked to salient events (e.g. war, new CEO or regulation)\footnote{\label{pxedrisk0}Some such risks are diversifiable and so unpriced; we consider those that are \emph{economically priced}: being assigned positive risk-premium moment by moment via recursive optimisation schemes.} that may alter some 'status quo' denoted by \(\{B=\mathbf{0}\}\) ('no change'). This setup mimics common real-world scenarios and the workhorse setting of model-risk studies, 
        {Markov hidden regime-switching}.
        
        If \(\mathbf{p}^B_0\) is known, all asset-pricing schemes perform \(B\)-inference under \emph{a priori} belief \(\pi^B_0=\mathbf{p}^B_0\). The more classical-leaning ones result in a \(B\)-risk aversion premium, and the more radical, in an additional ambiguity-aversion premium, by splitting the preference mechanism into one for \(B\)-risk and one for the rest. If \(\mathbf{p}^B_0\) is unknown, issues of bias \(|\pi^B_0-\mathbf{p}^B_0|>0\) 
        arise. Existing works, under stationary recurrence, either study the learning of \(\mathbf{p}^B_0\) itself (e.g. \cite{GTim}), or take \(\mathbf{p}^B_0\) as known (e.g. \cite{miao}), thereby allowing the identification of the \emph{ex-post} premia with the \emph{ex-ante}; this preclusion of bias is plausible, but no more so than inclusion (e.g. \cite{obe}).

Our one-off feature precludes \(\mathbf{p}^B_0\)-learning and amplifies the effect of bias \(|\pi^B_0-\mathbf{p}^B_0|>0\). It also decouples the pricing of \(B\)-risk from that of classical risk; the split leads to simplicity, not complexity, in contrast to existing approaches. This is reinforced further by a \emph{compound} information structure containing \emph{risk-specific} data, unlike the standard structure based solely on realised benefits (e.g. dividends), which are far less varied and frequent than real-world price-moving datasets. The one-off feature necessitates the analysis of market {cross-section} and dispersion\footnote{\label{dispersion}Interestingly, as a priced factor, dispersion has been linked to "structural change" and "the business cycle" (\cite{disp1}, \cite{disp2}).}, exposing \(B\)-risk as a major source of excess returns, which, as will be seen, have explicit expressions in terms volatility, risk-pricing and bias.  

We derive viable model-risk pricing that, under common economic conditions, 
takes an intuitive \emph{canonical} form, with a familiar price-of-risk and dynamic
. It enables an integrated analysis of \emph{ex-ante} {bias} and risk-premium, so that their \emph{ex-post} effects may be disentangled without ad hoc criteria
. As a result novel connections are made between a disparate group of phenomena and topics: Momentum (\cite{jt}), Low-Risk (\cite{vol-mo1}), Status Quo Bias (\cite{sq1}, \cite{sq2}) and Bewley's Inertia Axiom for decision making under Knightian uncertainty (\cite{Bew}).

\section{Setup}\label{su}

Consider, under {time-homogeneity}, 
asset benefit and price processes on \([0,\infty)\), in log-values, with time and risk-free discounts set to nil
.

\subsection{A Classical Starting Point: the Essential Standard Model}\label{assbasic}

For an asset-'\(b\)' offering \emph{cumulative} benefit \(\{Z_{t}(b)\}\), \(Z_{0}(b)=0\), with {independent increments} driven by standard Wiener noise \(\{dw_t\}\) on the usual basis 
\((\Omega:=C_0[0,\infty),\{\mathcal{F}_t\};\text{ }w)\), with \(\{\mathcal{F}_t\}\)-adapted, continuous and uniformly bounded drift \(\{r_t(b)\}\) and volatility \(\{\sigma_t\}\),\begin{align}\label{z0}
&dZ_t(b)={r}_{t}(b)dt+\sigma_t dw_t=:{r}_{t}(b)dt+dW|_t=:dW_b|_t
,\end{align}its usual viable pricing \(S^{b}_t\) at any \(t<\infty\) given \(Z_t({b})=z_t\) has the following form: 
\begin{align}
\label{S0}
&S^b_t=Y^b_t-RP^b(t)=z_t+\check{y}^b_\infty(t),\text{ with},\\\label{Y0}&Y^{b}_t
:=z_t+{E}_{W_b}[Z_\infty(b)-z_t|{\mathcal{F}_t}]=:
z_t+y^b_\infty(t),\\
&\label{RP0}{RP}^{b}(t):=Y^b_t-S^b_t=y^b_\infty(t)-\check{y}^b_\infty(t),\\\label{discountedexpectation}&\check{y}^b_\infty(t):={E}_{\check{W}_b}[Z_\infty(b)-z_t|{\mathcal{F}_t}],
\end{align}based on \emph{expected benefit} \(Y^{b}_{t}\), 
\emph{expected benefit growth} \(y^b_\infty(t):={E}_{W_b}[Z_\infty(b)-z_t|{\mathcal{F}_t}]<\infty\), and \emph{risk-premium} \({RP}^b(t)\ge0\), which is such that the \emph{risk-premium drift} \(\check{r}^b_tdt:=-d{RP}^b(t)\ge0\) well-defines \emph{RNE measure} \(\check{W}_{b}\sim W_{b}\) via \(d\check{W}_{b}|_t=-\check{r}^{b}_tdt+dW_b|_t\) and so \emph{RNE expected benefit growth} \(\check{y}^b_\infty(t)<\infty\). Note \(S^b_t={E}_{\check{W}_b}[Z_\infty(b)|{\mathcal{F}_t}]\), \(dS^b_t=\check{r}^b_tdt+dY^b_t\) and \(dY^b_t=dW|_t\)
. 
\subsection{Adding Model-risk the Standard Way}\label{addB}

Model-risk studies focus on ignorance of the drift, which, unlike volatility, is hard to measure (e.g. \cite{mert}). In its simplest form, the task is to choose from two alternatives. To be concrete, set \(\mathcal{B}\equiv\{+,-\}\), with \(r_{\Delta,t}:=r_t(+)-r_t(-)>0, \forall t<\infty\). Pricing under such a \(B\)-risk is guided by \(B\)-inference, based solely on the firmed-up benefit process
:\begin{equation}\label{z}
dZ_t(B)={r}_t({B})dt+dW|_t=:dW_{B}|_t,
\end{equation}whose \(\{B=b\}\)-version is (\ref{z0}), with 
\emph{\(B\)-sure asset-prices} \(\{S^B_t\}\) satisfying (\ref{S0}-\ref{discountedexpectation}). By observing the realised benefit process, one can detect the true drift eventually:\begin{align}\label{dYb(B)}
    &dY^b_t(B)=dZ_t(B)+dy^b_\infty(t)=sign(\overline{b})\mathbb{I}_{\{\overline{b}\}}(B)\cdot r_{\Delta,t} dt+dW|_t,\\
    \label{dSb(B)}&dS^b_t(B)=-dRP^b_t+dY^b_t(B)=\check{r}^b_tdt+sign(\overline{b})\mathbb{I}_{\{\overline{b}\}}(B)\cdot {r}_{\Delta,t} dt+
    dW|_t.
\end{align}
Wiener-drift testing (\cite{PShir} or Appendix \ref{infps}) is defined on the \emph{inferential stochastic basis} \(\big(\mathcal{B}\times{\Omega}
,\{\mathcal{F}_t\}
;\text{ }{\pi}^{B(\cdot)}_0{W}
_{B(\cdot)}(\cdot)\big)\), where \({\pi}^{B(\cdot)}_0\) is some unconditional law of \(B(\cdot)\in\mathcal{B}\)
. \emph{Regular tests} have \({W}_{+}\perp{W}_{-}\), \({W}_{+}|_{t}\sim{W}_{-}|_{t}\), \(\forall t<\infty\), and resolve \(B\) continuously as \(t\to\infty\). 

As only drift-difference counts, we de-trend the raw data to make \(\{Z_t(-)\}\) a martingale, give \(\{Z_t(+)\}\) the drift \(\{r^\Delta_t\}>0\), and so have \(dW_-|_t=\sigma_t dw_t\), 
\(y^-_\infty(t)=0\) and \(Y^-_t=
z_t\).

\subsection{Introducing Risk-Specific Data
}\label{Bdata}

Let there be \emph{\(B\)-specific} 
data \(\{{D}_t(B)\}\) 
\emph{independent} of \(\{Z_t(B)\}\) 
(and so {irrelevant} to \(B\)-sure asset-pricing). Define \(\{D_t(\pm)\}\) on \((\Omega
, \{\mathcal{F}^D_t\};\text{ }W^D_\pm)\), with \({W}^D_{+}\perp{W}^D_{-}\) and \({W}^D_{+}|_t\sim{W}^D_{-}|_t\), \(\forall t<\infty\), driven by some standard Wiener noise 
\(\{dw^D_t\}\), with \(\mathcal{F}^D_t\)-adapted continuous uniformly bounded drift \(\{r^D_{B,t}\}\) and volatility \(\{\sigma^D_t\}\); in de-trended form, with \(r^{D}_{\Delta,t}:=r^D_{t}(+)-r^D_{t}(-)>0\), \(\forall t<\infty\), it reads:
\begin{equation}\label{bD}
    dD_t(B)=\mathbb{I}_{\{+\}}(B)\cdot r^D_{\Delta,t}dt+\sigma^D_tdw^D_t=:dW^D_B|_t.
\end{equation}Drift-testing proceeds as before, with data \(\{\mathbf{D}_t\}:=\{(D_t,Z_t)\}\) and filtration \(\{\mathbf{F}_t\}:=\{\mathcal{F}^D_t\otimes\mathcal{F}_t\}\),
 and so an 'excessive' variability \emph{vs} tests based on benefit data \(\{Z_t\}\) alone. 

This structure adds modelling versatility; 
for instance it mimics a market 'obsessed' with a model-risk\footnote{The risk and risk-specific data may stem also from any not necessarily sound but self-consistent beliefs.}, whose economics dwarf routine ones and so whose truth is eagerly sought.


\section{The RNE Formulation of Model-Risk Pricing}\label{viablepx}
Asset-pricing under \(B\)-risk is based always on a reference \(B\)-belief process \(\{\pi^B_t\}\) generated on some inferential basis \(\big(\mathcal{B}\times{\mathbf{\Omega}},\{\mathbf{F}_t\};\text{ }\pi^B_0\mathbf{W}_{B}\big)\), \(\mathbf{{\Omega}}:=C_0[0,\infty)\times C_0[0,\infty)\), \(
{\mathbf{W}}_{B}:=
W^D_B\times W_B\). Viable asset-price \(\{S_t\}:=\{Y_t\}-\{RP_t\}\), as benefit expectation \(\{Y_t\}\) less risk-premium \(\{RP_t\}\), must have a martingale measure equivalent to that of \(\{Y_t\}\), a well-known \(\{\mathbf{F}_t\}\)-martingale under \(\pi^B_0\mathbf{W}_{B}\). 

Since \(B\)-sure risk-premia compensate for volatility, which is unaffected by \(B\)-risk ((\ref{z})), we shall presume \emph{identical \(B\)-sure risk-premia}\footnote{This 'purist' view is non-essential, but the resulting lighter notations promote clarity.}: \(RP^+(t)=RP^-(t)=:\check{RP}(t)\), and so \emph{identical risk-impact} \({X}_{\Delta}(t)
>0\) on benefit expectations and prices, as \({Y}^{\Delta}(t)-{S}^{\Delta}(t)\equiv RP^+(t)-RP^-(t)\), where \({Y}^{\Delta}(t):={{y}}^+_\infty(t)-{{y}}^-_\infty(t)>0\) and \({S}^{\Delta}(t):=\check{{y}}^+_\infty(t)-\check{{y}}^-_\infty(t)>0\).

\subsection{Viable Pricing in the Inferential Basis}\label{coreprop}

The key to pricing under \(B\)-risk is the risk-premium \(B\_RP_t\) attributable to \(B\)
. Note \({S}_t\in
[{S}^-_t,{S}^+_t]\) and so the existence of \emph{pricing coefficients} \(0\le\{{A}^B_t\}\le1\) such that \({S}_t\equiv\sum_B{{S}^B_t}{A}^B_t=:\langle{S}\rangle^A_t\). To summarise then: with \(\sigma^{\pi}_t:={(\pi^+_t{\pi}^-_t)^{\frac{1}{2}}}\) and in shorthand \(\langle \cdot\rangle^{\pi}_t:=\sum_B{{(\cdot)}^B}{\pi}^B_t\),\begin{align}
    \label{yt}  
&{Y}_t\equiv\langle Y\rangle^{\pi}_t=Y^-_t+X_\Delta(t)\cdot\pi^+_t,
\\
\label{abn}&{S}_t\equiv\langle{S}\rangle^A_t={S}^-_t+{X}_{\Delta}(t)\cdot A^+_t
,\\
\label{brpn}&{B\_RP}_t:=
\langle S\rangle^{\pi}_t-{S}_t={RP}_t-\check{RP}(t)={X}_{\Delta}(t)({\pi}^+_t-A^+_t),\\
\label{k}&k^A_t:=\frac{{B\_RP}_t}{{X}_{\Delta}(t)\sigma^{\pi}_t}=\frac{{\pi}^+_t-A^+_t}
{\sigma^{\pi}_t};
\end{align}note the standard deviation \({X}_{\Delta}(t)\sigma^{\pi}_t\) of \(B\_RP_t\) and so the \emph{price-of-\(B\)-risk} \(k^A_t\). 
Process \(\{k^A_t\}\) by definition is \(\{\mathbf{F}_t\}\)-adapted, continuous, and non-negative 
(at least \emph{ex ante})
.

It is convenient to work with \emph{\(B\)-risk only pricing} \(\langle{Y}\rangle^A_t\), ignoring \(B\)-sure risk-premium \(\check{RP}(t)\),\begin{align}\label{S0A}&
\langle{Y}\rangle^A_t:=Y_t-B\_RP_t=
S_t+\check{RP}(t)
= Y^-_t+{X}_{\Delta}(t)\cdot A^+_t
.
\end{align}Given tame \(\check{RP}(t)\), the \(B\)-risk only price \(\langle Y\rangle^{A}_t\) is viable \emph{iff.} asset-price \(S_t
=\langle Y\rangle^{A}_t-\check{RP}(t)\) itself is. A comparison of \(\langle Y\rangle^{A}_t\) 
with expected benefit \(\langle Y\rangle^{\pi}_t\) ((\ref{yt})) suggests (proof in Appendix \ref{proofp3.1}) that 
\(\langle Y\rangle^{A}_t\) is viable 
    \emph{iff.} the coefficients \(\{A^+_{t}\}\) amount to RNE \(B\)-risk beliefs\footnote{This is trivial if \(B\)-inference relies on \(\{D_t(B)\}\) only, 
    by the resulting independence of \(\{\pi^+_t(B)\}\) and \(\{Y^-_t(B)\}\).}: \(\{A^+_{t}\}=\{\hat{\pi}^+_t\}\sim\{\pi^+_t\}\).
    \begin{proposition}\label{p3}
Under \(B\)-risk, with data \(\{\mathbf{D}_t\}\) governed by reference law \(\pi^B_0\mathbf{W}_B\) over \([0,\infty)\), giving rise to \(B\)-risk beliefs \(\{\pi^B_t\}\), 
any viable asset-price \(\{S_t\}\) must have the form:\begin{align}\label{twintwin}
     &S_t=\langle{Y}\rangle_t^{\hat{\pi}}-\check{RP}(t),
 \end{align}where \(B\)-risk only pricing \(\langle{Y}\rangle_t^{\hat{\pi}}\) 
 has RNE law of the form \(\hat{\pi}^B_0(\hat{W}^D_B\times{W}_B)\), with \(\hat{W}^D_B\sim{W}^D_B\), \(\hat{\pi}^B_0\sim{\pi}^B_0\) and RNE \(B\)-risk beliefs \(\{\hat{\pi}^B_t\}\sim\{{\pi}^B_t\}\).
 \end{proposition}

The viable price dynamic, a \(\hat{\pi}^B_0(\hat{W}^D_B\times{W}_B)\)-martingale plus drift \(\{-d\check{RP}(t)\}\), thus reads:
\begin{align}\label{0'ri}d{S}_{t}(B)=&\check{r}_tdt+dY^-_t(B)+dX_{\Delta}(t)\Pi^+_t+X_\Delta(t)d\hat{\pi}^+_t\\\label{o'ri}
    =&[\check{r}_tdt+\sigma_tdw_{t}]+\\\label{b'unsureterms}
    +&[(\mathbb{I}_{\{+\}}-\hat{\pi}^+_{t})(B)\cdot r_{\Delta,t}dt]+\\\label{d'omi}
+&{X}_{\Delta}(t)\cdot(\sigma^{\hat{\pi}}_t)^2[(\mathbb{I}_{\{+\}}-\hat{\pi}^{+}_{t})(B)\cdot(\boldsymbol{{\sigma}}^{\mathbf{lD}}_t)^2dt+\boldsymbol{{\sigma}}^{\mathbf{lD}}_td\mathbf{\hat{w}^D}_{t}],\end{align}with \((\boldsymbol{{\sigma}}^{\mathbf{lD}}_t)^2:=({{\sigma}}^{l}_t)^2+({\sigma}^{lD}_t)^2\), \(\sigma^{l}_t:=\frac{r_{{\Delta,t}}}{\sigma_t}\), \(\sigma^{lD}_t:=\frac{r^D_{{\Delta,t}}}{\sigma^D_t}\), \(\boldsymbol{{\sigma}}^{\mathbf{lD}}_td\mathbf{\hat{w}^{D}}_{t}:={\sigma}^{l}_td{{{w}}_t}+{\sigma}^{lD}_td{\hat{w}^D_t}\) and \(d\hat{w}^D_t:=\frac{d\hat{W}^D_-|_t}{\sigma^D_t}\) (for details, see (\ref{findRNE}), Appendix \ref{proofp3.1}). It has three 'engines': the \(B\)-sure risk-premium drift (\ref{o'ri}), the \emph{model-drift} (\ref{b'unsureterms}) and the RNE \(B\)-inference process (\ref{d'omi}). 
\begin{rem}\label{rankorder}
    In practice \(B\)-risk data far exceed benefit data in variability, \(
    {(\sigma^{lD}_t)^2}\gg{(\sigma^{l}_t)^2}\), i.e. \((\boldsymbol{\sigma^{\mathbf{lD}}}_t)^2
    \approx(\sigma^{lD}_t)^2\); hence inference ((\ref{d'omi})) tends to dominate model-drift ((\ref{b'unsureterms})) and so price dynamic and dispersion (cf. footnote-\ref{dispersion}), even if the \(B\)-risk is unpriced (i.e. \(\hat{\pi}^+_t\equiv\pi^+_t\)).
\end{rem}

\subsection{Economic Considerations and Canonical Risk-Pricing}\label{keco}
Except the trivial, there have been no economic requirements so far. Economic asset-pricing relies on utility maximisation, which makes price-of-risk not only data-adapted but also a ({2-differentiable}) function of any reference belief process itself:\begin{equation}\label{kcon}
k^A_t\equiv \mathbf{k}^A(\pi^+_t).
\end{equation}Under FTAP-viability, this is a strong constraint, thanks to Ito's Formula (Appendix \ref{coro1proof}):

\begin{cor} \label{coro0}
    Viable pricing (\ref{twintwin}) subject to (\ref{kcon}) must have \(B\)-risk only pricing \(\langle{Y}\rangle_t^{{\Pi}}\) based on RNE \(B\)-risk beliefs \(\{{\Pi}^{B}_t\}\) under a canonical RNE law of the form \(\Pi^B_0\mathbf{W}_B\), \(\Pi^B_0\sim\pi^B_0\)
    . 
\end{cor}


\begin{cor} \label{coro1}Canonical 
price-of-\(B\)-risk \(k^{\Pi}_t:=\frac{{\pi}^+_t-\Pi^+_t}
{\sigma^{\pi}_t}\) 
has the following property:\begin{align}
\label{canon2}
        k^{\Pi}_t=(K^{\frac{1}{2}}-K^{-\frac{1}{2}})\sigma^{\Pi}_t,
        \end{align}where 
        \(K:={L^{\pi{+}}_0}/{L^{\Pi{+}}_0}
        ={L^{\pi{+}}_t}/{L^{\Pi{+}}_t} 
        \), with \(L^{\pi{+}}_{(\cdot)}:=\frac{\pi^+_{(\cdot)}}{
        \pi^-_{(\cdot)}}\) and \(L^{\Pi{+}}_{(\cdot)}:=\frac{\Pi^{+}_{(\cdot)}}{
        \Pi^-_{(\cdot)}}\), is a conserved quantity of the inference dynamic. Risk-pricing is positive iff. \(K>1\); unpriced \(B\)-risk has \(K=1\).
        \end{cor}

\begin{rem}\label{1-2}
      The gain-loss ratio of bearing \(B\)-risk given pricing \(\Pi^+_t\) is \(\frac{\Pi^-_{t}}{\Pi^{+}_{t}}\equiv\frac{1}{L^{\Pi{+}}_t}\). Parameter \(K\) thus scales up the 'plain' reference gain-loss ratio \(\frac{1}{L^{\pi{+}}_t}\) to 'sweeten' the risk-reward to offset any risk and ambiguity aversion. Indeed \(K\) is the gain-loss ratio itself at peak \(B\)-risk (where \({\pi}^+
      _t=\frac{1}{2}\) and \(L^{\pi{+}}_t=1\)). Naturally, buyers want \(K\gg1\), and sellers, \(K\ll1\); competition settles it ultimately. By (\ref{canon2}) the \(B\)-risk premium 
      has peak-value \(B\_RP^{max}[K]:=\frac{\sqrt{K}-1}{\sqrt{K}+1}X_{\Delta}(t)\) at \({\Pi}^+_{max}[K]=\frac{1
      }{\sqrt{K}+1}\) and so \({\pi}^+_{max}[K]=\frac{\sqrt{K}
      }{\sqrt{K}+1}\), whereas at peak \(B\)-risk, \(
      \sigma^{\pi}
      _t=\frac{1}{2}
      \), the \(B\)-risk premium is \(B\_RP_{\frac{1}{2}}[K]:=\frac{1}{2}\frac{K-1}{K+1}
      {X}_{\Delta}(t)\)
      . 
    We have \(B\_RP^{max}\approxeq B\_RP_{\frac{1}{2}}
    \) for \(K<2\), 
      and \({B\_RP^{max}}\sim{X}_{\Delta}(t)\sim{2}\times{B\_RP_{\frac{1}{2}}}\), with \({\Pi}^+_{max}\sim0\) and \({\pi}^+_{max}\sim1\), for \(K\gg1\). We surmise and assume henceforth that given 'normal' market competition, asset-pricing under \(B\)-risk satisfies 
      \({B\_RP^{max}}\approxeq{B\_RP_{\frac{1}{2}}}\ge0\) and so \(K\in[1,2)\).
      \end{rem}


\subsection{Connection with Familiar Price-Dynamics}\label{connect}
          
Canonical pricing is a \((\Pi^B_0\mathbf{W}^{T}_{B})\)-martingale with a 'tame tail', \(\check{RP}(t)\). Its price-dynamic is given by (\ref{0'ri}-\ref{d'omi}) with \(\hat{\pi}^+_t\equiv\Pi^+_t\), \(d\hat{w}^D_t\equiv dw^D_t\) and \(\boldsymbol{{\sigma}}^{\mathbf{lD}}_td\mathbf{{w}^{D}}_{t}:={\sigma}^{l}_td{{{w}}_t}+{\sigma}^{lD}_td{{w}^D_t}\), as if \(B\)-risk is unpriced, with \(B\)-inference \(\{d\pi^+_t\}\) and its RNE, \(\{d\Pi^+_t\}\), differing \emph{only} by an \emph{a priori} factor.  
The average price-dynamic with respect to \(B\)-risk under reference \(B\)-belief \(\pi^+_t\), ignoring the \(B\)-independent part, (\ref{o'ri}), and the dominated part, (\ref{b'unsureterms}), reads:\begin{equation}\label{klasic}
\boldsymbol\sigma_t\big[\frac{\boldsymbol\mu_t}{\boldsymbol\sigma_t}dt+d\mathbf{w^{D}}_{t}\big],\end{equation} where \(\boldsymbol\sigma_t:=X_{\Delta}(t)\cdot(\sigma^{{\Pi}}_t)^2\boldsymbol{\sigma}^{\mathbf{lD}}_t\), \(\boldsymbol{\mu}_t:=X_{\Delta}(t)\cdot(\pi^+_t-{\Pi}^+_t)(\sigma^{{\Pi}}_t\boldsymbol{\sigma}^{\mathbf{lD}}_t)^2\), and so ((\ref{canon2})):\begin{equation}\label{klasic'}
    \frac{\boldsymbol\mu_t}{\boldsymbol\sigma_t}
    =\frac{K^{\frac{1}{2}}-K^{-\frac{1}{2}}}{X_{\Delta}(t)}\cdot\frac{\sigma^{\pi}_t}{\sigma^{\Pi}_t}\cdot\boldsymbol\sigma_t=\frac{K-1}{X_{\Delta}(t)}\cdot\frac{\boldsymbol\sigma_t}{1+(K-1)\Pi_t};
\end{equation}by \((K-1)\)-expansion (Remark \ref{1-2}) the ubiquitous market-price-of-risk relationship emerges:
\begin{align}\label{capmlike}
    &\frac{\boldsymbol\mu_t}{\boldsymbol\sigma_t}\approx\bigg(\frac{K-1}{X_{\Delta}(t)}\bigg)\boldsymbol\sigma_t\text{, that is, }\boldsymbol\mu_t\approx\bigg(\frac{K-1}{X_{\Delta}(t)}\bigg)(\boldsymbol\sigma_t)^2.
\end{align}
.

\section{Bias and Risk-Pricing Measurement via Price Anomalies}\label{app}
\subsection{The Joint-Hypothesis Problem}\label{jhi}
Recall that our assets have model-risk \(B\) whose law \(\mathbf{p}^B_0\) may be unknown; they are classical otherwise, with known \(B\)-sure laws \(\mathbf{W}_B\) and economically priced risks (footnote-\ref{pxedrisk0}), \emph{ex ante}. 
The resulting \emph{ex-post} risk-premia \(\mathbf{rp}_t:=\langle Y\rangle^{\mathbf{p}}_t-S_t\), under the objectively true \(\{\mathbf{F}_t\}\)-conditional probabilities \(\{\mathbf{p}^B_t\}\) of \(B\)-outcomes
,  
satisfy
:\begin{equation}\label{ex-post-dft}
    \mathbf{rp}_t-\check{RP}(t)=B\_RP_t+
(\mathbf{p}^+_t-{\pi}^+_t){X}_{\Delta}(t),\text{ }t\le 1,\end{equation}
where the RHS are 'joint': any study of model-risk pricing \(B\_RP_t=(\pi^+_t-{\Pi}^+_t)X_\Delta(t)\) necessarily involves one on {bias} \((\mathbf{p}^+_t-{\pi}^+_t)X_\Delta(t)\), a joint-hypothesis problem
. 
Our approach to telling 'bias' from 'intentional risk-premia' relies on the latter obeying Corollary \ref{coro1} and Remark \ref{1-2}.

\subsection{The Risk of 'Change' from a 'Status Quo'}\label{strongstsq}
In case of \emph{rare} {hidden regime-switching}, so that each regime-span far exceeds the time needed for regime detection, consider the 'steady phase', when prices have long adjusted to the last regime-change, with RNE and reference beliefs in the current regime high, 
setting the scene for the detection of 'change', \(\{B=\mathbf{1}\}\), under \emph{a priori} belief \(\pi^\mathbf{1}_0\ll\frac{1}{2}\gg\Pi^\mathbf{1}_0\). Canonical prices with respect to '\(\mathbf{1}\)' read \(S^{{\Pi}}_t:=S^{\mathbf{0}}_t+sign(\mathbf{1})X_\Delta\cdot{\Pi}^{\mathbf{1}}_t
\) ((\ref{abn}) and Corollary \ref{coro0}), where \(sign(\mathbf{1})=\pm\) signals the economic nature of 'change'
, and risk-impact \(X_\Delta(t)\approxeq X_\Delta\) is presumed constant in effect during the active inference phase (short \emph{vs} regime-life).

\subsection{Bias and Risk-Pricing Parameters under Knightian Uncertainty}\label{app0}

Given any market-price implied RNE belief \({\Pi}^\mathbf{1}_t=\frac{|S^\Pi_t-S^{\mathbf{0}}_t|}{X_\Delta}\) and the existence of \emph{ex-ante}, intended, risk-pricing \(K\in[1,2)
\) (Remark \ref{1-2}), the reference belief \(
{\pi}^\mathbf{1}_t\) behind the given market-price can be found via \(
{L^{\pi{\mathbf{1}}}_t}:={L^{{{\Pi}}{\mathbf{1}}}_{t}}\cdot K^{sign[\mathbf{1}]}\), whose bias (if any
) is captured by parameter \(\rho\):\begin{align}
    \label{rho-K}&\rho:=\frac{L^{\mathbf{p}{\mathbf{1}}}_0}{L^{\pi{\mathbf{1}}}_0}=\frac{L^{\mathbf{p}{\mathbf{1}}}_t}{L^{\pi{\mathbf{1}}}_t},
    \text{ that is,}
    \\
\label{rhok}&\rho K^{sign[\mathbf{1}]}\equiv\frac{L^{\mathbf{p}{\mathbf{1}}}_0}{L^{{\Pi}{\mathbf{1}}}_0}=\frac{L^{\mathbf{p}{\mathbf{1}}}_t}{L^{{\Pi}{\mathbf{1}}}_t},
\end{align}with \(\rho=1\) for unbiased reference beliefs \(\pi^{\mathbf{1}}_0=\mathbf{p}^\mathbf{1}_0\)
. 
Only product (\ref{rhok}) is directly observable, so the challenge is to tell intention \(K\) from 'mistake' \(\rho\).

Our focus, given '{rare} switching' (Section \ref{strongstsq}), is the case of '{strong status quo}': high bias \(\rho\gg1\) under an objectively unlikely potential change, \(\mathbf{p}^\mathbf{1}_0<\frac{1}{2}\)
. Intriguingly, \emph{the inertia axiom} of \cite{Bew} for decision making under Knightian uncertainty states that status quo beliefs must be preserved {unless} abandonment is non-inferior in all states. It suggests in our context \emph{inertial pricing}, one with an extreme Status Quo Bias, \(\rho\to\infty\).

\subsection{Idealised Model-Risk driven Markets
}\label{app1}
Consider a standard market of classical assets with \(B\)-risks, 
where \emph{a)} the nature of the risks can vary across assets; \emph{b)} each \(B\)-outcome is realised individually, not necessarily independently, so that one event (e.g. war) can be a non-binary model-risk for the market overall; \emph{c)} 
\(B\)-specific data, via \(B\)-inference, dominate all else in driving asset dispersions (Remark \ref{rankorder}); 
\emph{d)} 
the \(B\)-risks are priced 
such that individual asset prices at any time are canonical up to deviations 
that can be treated as i.i.d noise across the market overall. 

\begin{rem}\label{noisecancellation}
     Viable pricing is anyway adjacent to canonical pricing (Item-\ref{tails}, Appendix \ref{infps}), and non-canonical viable prices can be seen as canonical ones based on noise- or bias-perturbed reference \(B\)-sure beliefs (Item-\ref{noisycan}\&\ref{sysdev}, Appendix \ref{weakp'}). Condition 'd)' then 
requires such perturbations to be non-systematic across the market.
\end{rem}

For brevity, 
assume also that in our market \emph{e)} all \(B\)-risk impacts, pricing and biases are identical up to noise that is independent of all else; \emph{f)} pricing and biases are symmetrical with respect to the direction of potential change, so that (\ref{rhok}) holds without \(K\) or \(\rho\) having any sensitivity to \(sign[\mathbf{1}]\); \emph{g)} the unconditional \emph{all-cause} probability of change is symmetrical with respect to the direction of change: at any time half of all potential changes are positive; for the 'uneven' case see Remark \ref{negmo} 
and footnote-\ref{uneven'}. 

\subsection{Conditional Cross-Sectional Trading}
 \subsubsection{Cohorts ranked by momentum}
Consider the trading of assets by cohorts defined by the condition \(\{{\Pi}^\mathbf{1}_t=v\}\), \(v\in(0,1)\). Under strong status quo, \((\sigma^{{\Pi}}_0)^2\ll\frac{1}{4}\), this amounts to ranking by \emph{momentum} (performance-to-date): \(\{{\Pi}^\mathbf{1}_t=v\}\) is associated with an excess-to-date of \(sign[\mathbf{1}]X_\Delta\cdot(v-{\Pi}^\mathbf{1}_0)\)
. Given \(sign[\mathbf{1}]=\pm\) and \({\Pi}^\mathbf{1}_t=v\) then, an asset whose \(B\)-risk pricing has parameter \(K\) and \(\rho\) offers a mean-excess-to-come of \(\mathbf{rp}(\pm,v):=\pm X_\Delta\cdot(\mathbf{p}^\mathbf{1}_t(v)-v)\), which reads ((\ref{rho-K}-\ref{rhok})): with \((\sigma^v)^2:=v\underline{v}:=v(1-v)\),
 \begin{equation}\label{mox}\mathbf{rp}(\pm,v)=\pm X_\Delta\cdot(\sigma^v)^2\cdot
    \frac{1-(\rho K^{\pm1})^{-1}}{v+\underline{v}(\rho K^{\pm1})^{-1}};
\end{equation}note that \(X_\Delta\cdot(\sigma^v)^2\) is essentially price-volatility ((\ref{d'omi}) and (\ref{klasic}-\ref{capmlike})). Under \(\rho=1\) the above is the \emph{ex-ante} risk-premium at \({\Pi}^\mathbf{1}_t=v\) ((\ref{canon2})): conditional trading earns what is fair when asset-pricing is unbiased. The excess due to \(B\)-risk, despite its unknown (and indeed unknowable if one-off) true rate \(\mathbf{p}^\mathbf{1}_t\), is written thus in quantities that are (more) knowable in principle.

For the market of Section \ref{app1}, the \(v\)-cohort average mean-excess-to-come given \(sign[\mathbf{1}]\) has the same form, in terms of some appropriate 'market-wide' risk-pricing and bias parameter. Henceforth, with parameter \(K\) and \(\rho\) understood as such wherever they appear, the cohort average, of form (\ref{mox}), shall be denoted by \(rp(\pm,v)\).

 
\subsubsection{Cohorts ranked by volatility}
Volatility-conditioned cohorts are of the form \(\{\sigma^{{\Pi}}_t=v\underline{v}\}\equiv\{{\Pi}^\mathbf{1}_t=v\}\cup\{{\Pi}^\mathbf{1}_t=\underline{v}\}\), \(v\in(0,\frac{1}{2}]\) ((\ref{d'omi}) and (\ref{klasic}-\ref{capmlike})). The \(v\underline{v}\)-cohort average mean-excess-to-come, denoted by \({rp}_t^\sigma(v)\), sums over the \(\{{\Pi}^\mathbf{1}_t=v\}\)- and \(\{{\Pi}^\mathbf{1}_t=\underline{v}\}\)-events as well as the directions \(sign[\mathbf{1}]\) of potential change
:\begin{equation}\label{volx}
    	rp_t^\sigma(v):=  \frac{\sum_{\pm}\mathbb{P}_t(\pm, v)\cdot rp(\pm,v)+\sum_{\pm}\mathbb{P}_t(\pm,\underline{v})\cdot rp(\pm,\underline{v})}{\mathbb{P}_t(v)+\mathbb{P}_t(\underline{v})},\end{equation}in terms of the probabilities \(\mathbb{P}_t(\pm,\cdot):=\mathbb{P}_t\big(sign[\mathbf{1}]\cap\{{\Pi}^\mathbf{1}_t=\cdot\}\big)
        \) and \(\mathbb{P}_t(\cdot)=
        \mathbb{P}_t(+,\cdot)+\mathbb{P}_t(-,\cdot)\), 
        which can be found explicitly for Wiener-drift tests (Appendix \ref{momixfun}); their ratios corroborate the association of high volatility with negative momentum, a well-known empirical phenomenon (e.g. \cite{vol-mo1} and \cite{vol-mo}).     

\subsection{The Momentum and Low-Risk Effects of Status Quo Bias}\label{moeff}
The risk-rewards (\ref{mox}-\ref{volx}) of conditional trading are both concave in \(v\in(0,1)\), driven essentially by price-volatility \(\propto(\sigma^v)^2\) and vanishing as \((\sigma^v)^2\to0\).

Momentum trading, ranking assets by excess-to-date and profiting from their differing average excess-to-come, exploits the Momentum effect: excess-to-date \(\pm(v-{\Pi}^\mathbf{1}_0)X_{\Delta}\) 
tends to persist and the effect grows with bias \(\rho>1\) ((\ref{mox})), vanishing as \(\rho\to1\). 
It has peak profitability \(\frac{1}{2}(rp^{+mo}_{max}-rp^{-m o}_{max})\), by being long the \(v^{+mo}_{max}\)-cohort and short the \(v^{-mo}_{max}\)-cohort, where
:\begin{equation}\label{vmo}
    v^{\pm mo}_{max}:=\frac{1}{(\rho K^{\pm1})^{\frac{1}{2}}+1},\text{  }rp^{\pm mo}_{max}:=\pm \frac{(\rho K^{\pm1})^\frac{1}{2}-1}{(\rho K^{\pm1})^\frac{1}{2}+1}X_\Delta.
\end{equation}Momentum-reward thus rises with conditioning-momentum, up to a point, 
beyond which it declines, 
as recorded in \cite{vol-mo1}. The size of momentum-reward rises with bias \(\rho\) and has a maximum of \(X_\Delta\) as \(\rho\to\infty\).

Volatility trading (\ref{volx}) exhibits a Low-Risk effect (visible in (\ref{vmo}) for momentum also), which inverts the classical dictum 'high risk high reward'. The following solution (Appendix \ref{locale}) for the peak-location \(v^\sigma_{max}\) and -size \(rp^\sigma_{max}\) of (\ref{volx}) is exact at \(\rho=1\) (no bias) and valid to leading order 
at \(\rho\gg1\) (high bias)
:\begin{equation}\label{lowrisk}
    v^\sigma_{max}:=\frac{1}{\rho+1},\text{  }rp^\sigma_{max}:=\frac{1}{2}\frac{K-1}{K+1}X_\Delta.
\end{equation}

At \(\rho=1\), we have \(v^\sigma_{max}=\frac{1}{2}\) (peak price-volatility), thus confirming 'high risk high reward'
. Under bias, \(\rho>1\), there is a Low-Risk effect: rewards peak at \(v^\sigma_{max}\approx\rho^{-1}\), and 
peak-risk pays poorly, \(rp^\sigma(\frac{1}{2})=\frac{X_\Delta}{2}\mathcal{O}(\frac{K-1}{\rho}-1)\). Note the separation: peak-location \(v^\sigma_{max}\) reveals bias \(\rho\), while peak-size \(rp^\sigma_{max}\), the reward intended for peak-risk (Remark \ref{1-2}), reveals risk-pricing \(K\). 
\subsection{Discussion}\label{con-dis}


Relationship (\ref{vmo}-\ref{lowrisk}) apply in markets exposed to concurrent \(B\)-risks that may differ in nature but enjoy similar (enough) market risk-pricing \(K\in[1,2)\) and bias \(\rho\gg1\)
. Idealisations '\emph{a-d)}' and simplifications '\emph{e-g)}' notwithstanding, the concave dependence of cross-sectional excess on price-volatility 
is likely a robust feature of situations where model-risk inference dominates asset-price movements and dispersions ((\ref{d'omi}) and Remark \ref{rankorder}).

The analysis demonstrates how model-risks whose laws are unknown (or unknowable) may reveal themselves through exploitable excess returns that can be expressed in variables that are more readily observed and measured. We do not identify the \emph{ex-post} with the \emph{ex-ante} in the presence of model uncertainty, unlike most existing works on the subject. Instead, potential biases are allowed, and ways to separate bias effect from that of \emph{ex-ante} risk-pricing without the usual joint-hypothesis issues, suggested. Although the explicit results derived depend on a series of (over)simplifications, the idea of using the necessary economic conservatism of risk-pricing and the lack thereof of bias (to tell them apart) is likely valid in general.

\begin{appendix}
\section{Properties of Binary Inferential Testing}\label{infps}
\begin{enumerate}

\item\label{takinglimit}\emph{Regular Continuous-Time Inference.} The sequential testing of given measure-pair \({W}
    _{\pm}\) relies on (log) likelihood-ratios (log-LR) \(l^{{+}}_t:=\log L^{{+}}_t:=\log\frac{d{W}
    _{+}|_t}{d{W}
    _{-}|_t}\), to which Bayes' Rule applies: 
    \(l^{{+}}_{s+t}=l^{{+}}_{s}+l^{{+}}_{t}\), \(\forall s,t\in\mathbb{R}^+\).
    For homogeneous data \(d{W}_B|_\tau=\mathbb{I}_{\{+\}}(B)r_\Delta d\tau +\sigma dw_\tau\), \(B\in\mathcal{B}:=\{+,-\}\), driven by standard Wiener noise \(w\) on \({\Omega}:=C_0[0,\infty)\), the log-LR process \(l^{{+}}_\tau\) obeys the following (\cite{PShir}): in \emph{signal-to-noise} \(\sigma^l:=\frac{r_\Delta}{\sigma}\),\begin{equation}\label{w}dl^{{+}}_\tau(B)=(-1)^{\mathbb{I}_{\{-\}}(B)}\frac{(\sigma^l)^2}{2}d\tau + \sigma^l dw_\tau=\frac{dL^{{+}}_\tau}{{L^{{+}}_\tau}}(B)-\frac{(\sigma^l)^2}{2}d\tau.\end{equation} 

\item{\emph{Change of Clock.}\label{ckchng} A suitable \emph{time-change} (\cite{PShir}), \(
\tau\mapsto{t}
\), can result in time-varying data with continuous uniformly bounded 
drift \(\{r_{{\Delta,t}}\}\) and volatility \(\{\sigma_t\}\) with respect to some time-\(t\) standard Wiener noise. The test dynamics, written in terms of signal-to-noise \(\sigma^l_t:=\frac{r_{{\Delta,t}}}{\sigma_t}\), then read:
\begin{align}
    \label{w'}
&dl^{{+}}_t(B)=
(-1)^{\mathbb{I}_{\{-\}}(B)}\frac{(\sigma^{l}_t)^2}{2}dt+ \sigma^{l}_tdw_t=\frac{dL^{{+}}_t}{L^{{+}}_t}(B)-\frac{(\sigma^l_t)^2}{2}dt;
\\\label{rawdrift}  
&\frac{d\pi^+_{t}}{(\sigma^\pi_t)^{2}}(B)= \big(\mathbb{I}_{\{+\}}-\pi^+_t\big)(B)\cdot(\sigma^l_t)^2dt+\sigma^l_tdw_t,
\end{align}
where \(\{\pi^{+}_t\}\) is the resulting \(B\)-beliefs, 
a martingale on \(\big(\mathcal{B}\times\Omega,\{\mathcal{F}_{t}\};\text{ }\pi^B_0\times W_B\big)\).
}

\item{\label{convergent}\emph{Inferential Resolution.} We have \(|l^{{+}}_t|
<\infty\) almost surely under equivalence \({{W}}_{+}|_t\sim{W}_{-}|_t\) (Radon-Nikodym Theorem). For Wiener-drift tests (\ref{w}-\ref{w'}), it means, almost surely,
\begin{equation}\label{resovlingls'}
\int_0^t(\sigma^l_u)^2du\equiv 
\int_0^{t}\big(\frac{r_{\Delta,u}}{\sigma_u}\big)^2du<\infty.
\end{equation} 
Such tests resolve, that is, \(\lim_{t\to \infty}|l^{{+}}_t|\to\infty\), 
\emph{iff.} \({{W}}_{+}\perp{W}_{-}\): almost surely,
\begin{equation}\label{resovlingls}\lim_{t\to \infty}
\int_0^{t}(\sigma^l_u)^2du\equiv\lim_{t\to \infty}
\int_0^{t}\big(\frac{r_{\Delta,u}}{\sigma_u}\big)^2du=\infty.
\end{equation}Resolution thus 
precludes martingale measures for the associated test processes.
}
\item{\label{tails}\emph{Adjacency.} Resolving and thus divergent log-LR processes, \(\{l^{{+}}_t\}\) and \(\{\hat{l}^{{+}}_t\}\), with respective test measure-pair \({W}_\pm\) and \(\hat{{W}}_\pm\) for which equivalence \({{W}}_{\pm}\sim \hat{W}_\pm\) holds, 
can differ at most by a finite amount almost surely (Radon-Nikodym Theorem): at any \(t<\infty\), 
we have,\begin{equation}\label{logtrick}|\hat{l}^{{+}}_t-l^{{+}}_t|=\bigg|\log{\frac{d\hat{{W}}_{+}|_t}{d{W}_+|_t}}-\log{\frac{d\hat{{W}}_{-}|_t}{d{W}_{-|_t}}}\bigg|<\infty;\end{equation}further, the RHS remains almost surely bounded as \(t\to \infty\), given \({{W}}_{\pm}\sim \hat{W}_\pm\). Denote adjacency by \(\{\hat{l}^{{+}}_t\}\approx\{l^{{+}}_t\}\).
}

\item\emph{Wiener Adjacency and Equivalence.}\label{tailsame} Consider a Wiener test measure-pair \(\hat{{W}}_\pm\) \emph{vs} the original \({W}_\pm\), with \(d\hat{{W}}_\pm|_t=-\frac{1}{2}\hat{\theta}_{\pm,t}dt+d{W}_\pm|_t\), so that \(\hat{r}_{\Delta,t}={r}_{\Delta,t}-\frac{1}{2}\hat{\theta}_{\Delta,t}\). Its log-LR dynamic \(d\hat{l}^{{+}}_t\) has form (\ref{w'}), in terms of \(d\hat{w}_t:=\frac{d\hat{W}_-|_t}{\sigma_t}\) and \(\hat{\sigma}^l_t:=\frac{\hat{r}_{\Delta,t}}{\sigma_t}\). 
 If the measure generated by \(\{d\hat{l}^{{+}}_u\}|_t\), \(t<\infty\), is equivalent to that by \(\{d{l}^{{+}}_u\}|_t\), write \(\{\hat{l}^{{+}}_u\}|_t\sim\{{l}^{{+}}_u\}|_t\) and call them equivalent. Such equivalence in the Wiener setting demands \(\hat{\sigma}^l_u=\sigma^l_u\); that is, \(\{\hat{\theta}_{\Delta,u}\}=0
 \) and \(d\hat{{W}}_{\pm}|_t=-\frac{1}{2}\hat{\theta}_{t}dt+d {W}_\pm|_t\) for some 
 \(\{\hat{\theta}_{t}\}\). Then, by \((\frac{\hat{\theta}_t}{r_{{\Delta},t}}\cdot\sigma^l_t)^2\equiv(\frac{\hat{\theta}_t}{\sigma_t})^2\),\begin{equation}\label{lhat=l+drft}E |\hat{l}^{{+}}_t-l^{{+}}_t|=
 \frac{1}{2}E\int_0^t\frac{|\hat{\theta}_u|}{r_{\Delta,u}}(\sigma^l_u)^2du\equiv
 \frac{1}{2}E\int^t_0\frac{|\hat{\theta}_u|\cdot r_{\Delta,u}}{(\sigma_u)^2}du,\end{equation}where \(E\) denotes expectation 
 under \emph{either} law \(W_\pm\). Thus, Novikov's Condition is met by the relative drift 
of \(\{\hat{l}^{{+}}_t\}\) \emph{vs} \(\{l^{{+}}_t\}\) \emph{iff.} it is met by that 
of \(\hat{W}_\pm|_t\) \emph{vs} \(W_\pm|_t\); so on any \([0,t]\), \(t<\infty\), the relationships below are equivalent:\begin{equation}\label{neat}\hat{W}_\pm|_t\sim W_\pm|_t\text{ }\Longleftrightarrow\text{ }\{\hat{l}^{{+}}_u(\pm)\}|_t\approx\{{l}^{{+}}_u(\pm)\}|_t\text{ }\Longleftrightarrow\text{ }\{\hat{l}^{{+}}_u(\pm)\}|_t\sim\{{l}^{{+}}_u(\pm)\}|_t.\end{equation}

\item{\label{unique}\emph{Extension as \(t\to\infty\).} 
Recall that no resolving log-LR processes can have martingale measures as \(t\to \infty\) (Item-\ref{convergent}), so the last part of (\ref{neat}) has no meaning in that case. 
Adjacency into risk-resolution 
does ((\ref{logtrick}-\ref{lhat=l+drft})), provided:\begin{align}\label{finitereldrft}&\lim_{t\to \infty}E\int^{t}_0\frac{|\hat{\theta}_u|}{r_{\Delta,u}}(\sigma^l_u)^2du=\lim_{t\to \infty}E\int^{t}_0\frac{|\hat{\theta}_u|\cdot r_{\Delta,u}}{(\sigma_u)^2}du<\infty,\text{ and so,}\\\label{novs}&\lim_{t\to \infty}E\int^{t}_0\big(\frac{\hat{\theta}_u}{r_{\Delta,u}}\big)^2(\sigma^l_u)^2du=\lim_{t\to \infty}E\int^{t}_0\big(\frac{\hat{\theta}_u}{\sigma_u}\big)^2du<\infty,\end{align}
automatically true for log-LR pairs arisen from equivalent test measure-pairs \(\hat{W}_\pm\sim W_\pm\) already given on \(C_0[0,\infty)\).
}
\end{enumerate}

\section{Proposition \ref{p3}, Corollary \ref{coro0} and NUPBR-Viability}\label{proofp3} 
\subsection{Viable Pricing in the Inferential Basis}\label{proofp3.1} Recall that it suffices to prove it for \(B\)-risk only pricing 
\(\langle{Y}\rangle^A_t=
S_t+\check{RP}(t)=Y^-_t+X_\Delta(t)\cdot A^+_t\). 
Note that the expected benefit process \(
\langle{Y}\rangle^\pi_t=Y^-_t+X_\Delta(t)\cdot \pi^+_t\) ((\ref{yt})) is already a martingale on \(\big(\mathcal{B}\times{\mathbf{\Omega}},\{\mathbf{F}_t\};\text{ }\pi^B_0\mathbf{W}_B\big)\), with \(\mathbf{{\Omega}}:=C_0[0,\infty)\times C_0[0,\infty)\) and \(\{\mathbf{F}_t\}:=\{\mathcal{F}^D_t\otimes\mathcal{F}_t\}\), given \(B\)-conditionally independent joint-dataflow \(\{(D_t, Z_t)\}\) under law \(\pi^{B}_0{\mathbf{W}}_{B}:=\pi^{B}_0(W^D_B\times W_B)\).

FTAP-viability in the Wiener setting requires \(d\langle{Y}\rangle^\pi_t\) and \(d\langle{Y}\rangle^A_t\) to differ only by drifts. Any RNE measure \(\hat{\pi}^B_0\hat{\mathbf{W}}_{B}\) thus must preserve the \(B\)-conditional independence above, and so, for \(B\)-risk only pricing, must have the form \(\hat{\mathbf{W}}_{B}=\hat{W}^D_B\times {W}_B\), \(\hat{W}^D_B\sim{W}^D_B\).

For term-matching between 
  \(dA^+_t\) and \(d\pi^+_t\)
  , it is easier to work with \(\hat{l}^{{+}}_t:=\log\frac{A^+_t}{A^-_t}=:\log \hat{L}^{{+}}_t\), well-defined given \(0<\{A^+_t\}<1\) up to any \(t<\infty\)
  . The dynamics of \(d\hat{l}^{{+}}_t\) and \(dA^+_t\) have the form and relationship of (\ref{w'}) and (\ref{rawdrift}) respectively. As the reference log-LR process is a sum \({l}^{{+}}_t={l}^{D{+}}_t+{l}^{Z{+}}_t\) of two independent ones adapted respectively to \(\{\mathcal{F}^D_t\}\) and \(\{\mathcal{F}_t\}\), 
we must have, under viability, \(\hat{l}^{{+}}_t=\hat{l}^{D{+}}_t+\text{ }{l}^{Z{+}}_t\), with the second part unmodified for \(B\)-risk only prices, 
and the first, equivalent to \(
\{{l}^{D{+}}_u\}|_{t<\infty}\) (Item-\ref{tailsame}, Appendix \ref{infps}): meaning 
\(\exists\text{ }\{\hat{\theta}^{D}_u\}\), \(B\)-independent, such that 
up to any \(t<\infty\), with \(d\hat{w}^D_t:=-\frac{1}{2}\frac{\hat{\theta}^{D}_{t}}{ r^{D}_{\Delta,t}}\cdot\sigma^{lD}_tdt+dw^D_t=:\frac{d\hat{W}^{D}_-|_t}{\sigma^D_t}
\), the following hold
:\begin{align}\label{findRNE}
d\hat{W}^{D}_\pm|_t&=-\frac{1}{2}\hat{\theta}^{D}_tdt+dW^D_\pm|_t\sim dW^D_\pm|_t,\\    
    \label{w''}
d\hat{l}^{D{+}}_t(B)&=\frac{1}{2}(-1)^{\mathbb{I}_{\{-\}}(B)}(\sigma^{lD}_t)^2dt+\sigma^{lD}_td\hat{w}^D_t.
\end{align}
All remain well-defined as \(t\to\infty\) (
Item-\ref{unique}, Appendix \ref{infps}); hence Proposition \ref{p3}. \(\blacksquare\)

\subsection{Canonical Viable Pricing in the Inferential Basis}\label{coro1proof}

Under Proposition \ref{p3}, condition (\ref{kcon}) demands \(\hat{l}^{{+}}_t(\pm)= g[{l}^{{+}}_t(\pm)]\), with \(g\) twice-differentiable. Both \(d{l}^{{+}}_t\) and \(d\hat{l}^{{+}}_t\) have dynamics (\ref{w'}). 
Given viability, only their drifts can differ, which means \(g'=1\) (due to the term \(g'\cdot d{l}^{{+}}_t\) in Ito's Formula); 
hence Corollary \ref{coro0}.

In the context of \(B\)-risk only pricing (Appendix \ref{proofp3.1}), by integrating (\ref{w''}) we have the following gauge for the 'RNE distortion' of the reference belief process up to any \(t<\infty\):\begin{equation}\label{corolaw}
\hat{\kappa}_t:=\bigg(\frac{L^{{+}}_t}{L^{{+}}_0}\bigg)\bigg/\bigg(\frac{\hat{L}^{{+}}_t}{\hat{L}^{{+}}_0}\bigg)=
\exp
{\frac{1}{2}\int_0^{t}\frac{\hat{\theta}^{D}_{u}}{r^{D}_{\Delta,u}}\cdot(\sigma^{lD}_u)^2du}
<\infty;\end{equation}it remains bounded almost surely as \(t\to\infty\), given viability. 
Corollary \ref{coro1} then follows, as canonical pricing by definition has \(\hat{\theta}^{D}_{u}\equiv0\) and so \(\hat{\kappa}_t\equiv1\). \(\blacksquare\)

\subsection{Biased \(B\)-sure Beliefs and Bounded Total Bias}\label{weakp'}

As shown, only canonical prices satisfy economic condition (\ref{kcon}). Markets with many such assets and possibly heterogenous agents however may not abide by it at the individual asset level, so that viable individual asset-price dynamic (\ref{0'ri}-\ref{d'omi}) ((\ref{findRNE}-\ref{w''})) that is non-canonical cannot be dismissed. On the other hand, such prices have the following interpretations:

\begin{enumerate}

 \item \emph{Canonical pricing perturbed by noise.}\label{noisycan} 
 Any given viable price process has at any instant \(t<\infty\) an implied \(B\)-risk pricing parameter \(\hat{K}_t\) via (\ref{corolaw}): with \(K:=\frac{L^{{+}}_0}{\hat{L}^{{+}}_0}\),\begin{equation}\label{cumulativeeffect}
     \hat{K}_t:=
     \frac{L^{{+}}_t}{\hat{L}^{{+}}_t}=\hat{\kappa}_t\times K
     .
 \end{equation}If \(\{\hat{\kappa }_{t}\}\approx1\) and \(\{E\text{ }\hat{\kappa }_{t}\}\approx1\), such pricing is essentially canonical but 'noisy'.
 
 \item\emph{Canonical pricing with biased \(B\)-sure beliefs.}\label{sysdev} If \(\hat{\kappa}_t\) deviates systematically from \(1\), 
 it is more plausibly due to biased reference \(B\)-sure beliefs (e.g. of the form (\ref{findRNE})) than the violation of basic principles (e.g. (\ref{kcon})). 
Viability requires such bias to be \(B\)-independent (Item-\ref{tailsame}, Appendix \ref{infps}), and its cumulative effect, bounded ((\ref{corolaw}))
. 

\end{enumerate}

\section{THE MOMENTUM AND LOW-RISK FORMULAE}\label{moloder}
For the assets of Section \ref{app0}-\ref{app1}, all the distributions required derive from the underlying log-LR process \(\{l^{\mathbf{1}
}_t\}\), 
whose level \(l^{\mathbf{1}
}_t\) at any \(t<\infty\) follows the Normal distribution:\begin{equation}\label{gauss}
\mathcal{N}\bigg(\boldsymbol{\mu}^{\mathbf{lD}}_{[t]}(B):=\frac{(-1)^{B+1}}{2}(\boldsymbol{\sigma}^{\mathbf{lD}}_{[t]})^2,\text{ } (\boldsymbol{\sigma}^{\mathbf{lD}}_{[t]})^2:= \int^t_0(\boldsymbol{\sigma}^{\mathbf{lD}}_{u})^2du\bigg),\text{ }B\in\{\mathbf{0},\mathbf{1}\}.
\end{equation}The distribution of RNE beliefs \(\Pi^\mathbf{1}_t\) and so prices then derive from Bayes' Rule by change of variable: for \(\Pi^\mathbf{1}_t=v\) given \(l^{\mathbf{1}
}_t=l\), we have \(dl=(v\underline{v})^{-1}dv\), with \(\underline{v}:=1-v\), as ((\ref{rhok})):\begin{align}\label{chngvarb}
&l=H_{\Pi\pm}+\log\frac{v}{\underline{v}},\text{ where}\\
&H_{\Pi\pm}:=\log\frac{1-{\Pi}^\mathbf{1}_{0\pm}}{{\Pi}^\mathbf{1}_{0\pm}}\equiv\log \frac{1-{\mathbf{p}}^\mathbf{1}_0}{{\mathbf{p}}^\mathbf{1}_0}+\log \rho K^{\pm1}.
\end{align}
Note that \(H_{\Pi\pm}\) is the log-LR hurdle for \(\Pi^\mathbf{1}_t\ge\frac{1}{2}\), and likewise, \(H_{\mathbf{p}}:=\log \frac{1-{\mathbf{p}}^\mathbf{1}_0}{{\mathbf{p}}^\mathbf{1}_0}\), for \(\mathbf{p}^\mathbf{1}_t\ge\frac{1}{2}\). Our focus, 'strong status quo' (Section \ref{app0}), corresponds to 
\(H_{\Pi\pm}\gg H_{\mathbf{p}}>0\).

\subsection{Tracking Inferential Progress and Bias Dominance}
The objective degree of \(B\)-certainty
, given (\ref{gauss}), is tracked by \(|C^{\mathbf{p}}_{t}(B)|\): with 
\(t_{\mathbf{p}/t}:=\frac{2H_{\mathbf{p}}}{(\boldsymbol{\sigma}^{\mathbf{lD}}_{[t]})^2}\),\begin{align}
C^{\mathbf{p}}_{t}(B):=\frac{H_{\mathbf{p}}-\boldsymbol{\mu}^{\mathbf{lD}}_{[t]}(B)}{\boldsymbol{\sigma}^{\mathbf{lD}}_{[t]}}=\frac{1}{2}\boldsymbol{\sigma}^{\mathbf{lD}}_{[t]}\cdot\big({(-1)^{B}}+
{t_{\mathbf{p}/t}}
\big).
\end{align}Certainty is high under \(\boldsymbol{\sigma}^{\mathbf{lD}}_{[t]}\ll H_{\mathbf{p}}\) (meagre data) or \(\boldsymbol{\sigma}^{\mathbf{lD}}_{[t]}\gg H_{\mathbf{p}}\) (rich data). It bottoms out when \(H_{\mathbf{p}}=\boldsymbol{\mu}^{\mathbf{lD}}_{[t]}(B)\) has a solution, that is, when \(
(-1)^{B+1}=t_{\mathbf{p}/t}\)
.

For the RNE inference process \(\{\Pi^\mathbf{1}_t\}\), its tracker \(|C^{\Pi}_{t\pm}(B)|\) parallels the above: 
\begin{align}
C^{\Pi}_{t\pm}(B):=&\frac{H_{\Pi\pm}-\boldsymbol{\mu}^{\mathbf{lD}}_{[t]}(B)}{\boldsymbol{\sigma}^{\mathbf{lD}}_{[t]}}=\frac{1}{2}\boldsymbol{\sigma}^{\mathbf{lD}}_{[t]}\cdot\big({(-1)^{B}}+
{t_{\Pi\pm/t}}
\big)\\
=&C^{\mathbf{p}}_{t}(B)+\frac{1}{2}\boldsymbol{\sigma}^{\mathbf{lD}}_{[t]}\cdot(
{t_{\rho/t}}
\pm
{t_{K/t}}
),
\end{align}with \(t_{\rho/t}:=\frac{2\log\rho}{(\boldsymbol{\sigma}^{\mathbf{lD}}_{[t]})^2}\), \(t_{K/t}:=\frac{2\log K}{(\boldsymbol{\sigma}^{\mathbf{lD}}_{[t]})^2}\) and \(t_{\Pi\pm/t}=t_{\rho/t}\pm t_{K/t}+t_{\mathbf{p}/t}\); recall \(K\in[1,2)\) (Remark \ref{1-2}) and note \(t_{\rho/t}\gg t_{K/t}>0\). 

Bias creates an extra 'burden of proof' to overcome, the extent of which is tracked by \(t_{\rho/t}\). Exploitable excess-profits occur when the objective conditional probability of change becomes meaningful while the bias-induced 'burden of proof' is still high: \begin{equation}\label{window}\big{\{}
 {t_{\mathbf{p}/t}}\ll1 \big{\}}
\text{ }\bigcap\text{ }\big{\{}
t_{\rho/t}\gg1\big{\}};\end{equation}that is, data becoming sufficient to overcome the objective hurdle, but not yet the bias.

\subsection{Momentum and Volatility Mixtures}\label{momixfun}
The probability behind mean-excess (\ref{volx}) obeys the basic law of conditioning,\begin{align}\label{lawP}
\mathbb{P}_t(\pm,v)\equiv\mathbb{P}_t(
v)\mathbb{P}_t(\pm|
v)\equiv\mathbb{P}_t(\pm)\mathbb{P}_t(
v|\pm).
\end{align}The ratio of \(\{\Pi^\mathbf{1}_t=\underline{v}\}\)- \emph{vs} \(\{\Pi^\mathbf{1}_t=v\}\)-events given \(sign[\mathbf{1}]\) 
reads: with \(\mathbb{P}_t(
v|\pm)\) by (\ref{gauss}) via (\ref{chngvarb}),\begin{equation}\label{rv}
R_{t|\pm}(\underline{v}/v)(B):=\frac{\mathbb{P}_t(\underline{v}|\pm)}{\mathbb{P}_t(v|\pm)}(B)=({\underline{v}}/{v})^{-(t_{\Pi\pm/t})-(-1)^B}
,\text{ }v\in(0,\frac{1}{2}],\text{ }B\in\{\mathbf{0},\mathbf{1}\}
.
\end{equation}Note the dominance of \(\{\Pi^\mathbf{1}_t=v\}\)-events under \(t_{\rho/t}\gg 1\) irrespective of \(B\)-value.

To compute (\ref{volx}), we also need the mix-ratio \(M_{t|v}:={\mathbb{P}_t(+|v)}/{\mathbb{P}_t(-|v)}\) between members with \(\{sign[\mathbf{1}]=+\}\) and those with \(\{sign[\mathbf{1}]=-\}\) within each event-set \(\{\Pi^\mathbf{1}_t=v\}\), \(v\in(0,1)\). Given (\ref{lawP}), we have \(M_{t|v}= {\mathbb{P}_t(v|+)}/{\mathbb{P}_t(v|-)}\), as \(\mathbb{P}_t(\pm)=\frac{1}{2}\) by setup. Therefore:
\begin{align}\label{momix}
M_{t|v}(B)
&= 
(\frac{\rho v}{\underline{v}})^{-t_{K/t}}K^{-(t_{\mathbf{p}/t})-(-1)^B},\text{ }B\in\{\mathbf{0},\mathbf{1}\};\end{align}
and likewise the volatility-conditioned mix-ratio \(M_{t|v\underline{v}}:={\mathbb{P}_t(+|v\underline{v})}/{\mathbb{P}_t(-|v\underline{v})},\text{ }v\in(0,\frac{1}{2}]\):\begin{equation}\label{volmix}
M_{t|v\underline{v}}= M_{t|v}\frac{1+R_{t|+}(\underline{v}/v)}{1+R_{t|-}(\underline{v}/v)}.
\end{equation}The source of mix conditionality is risk-pricing \(K\): unconditional (even) mix prevails if \(K=1\).

\begin{rem}\label{negmo}
    Mix-ratios 
    are monotone-declining in \(v\), with \(M_{t|\frac{1}{2}}= M_{t|\frac{1}{2}\frac{1}{2}}\approx\rho^{-t_{K/t}}
    =e^{{-t_{\rho/t}\cdot\log K}}\) at \(v=\frac{1}{2}\), peak volatility. If \(\rho\gg1\), in the window (\ref{window}) of excess-profit opportunity, observed mix is highly negative at high volatility, reverting to its unconditional level as \(B\)-risk resolves. An uneven unconditional mix merely introduces a {constant factor} in all above.
\end{rem}


\subsection{Peak-Reward Location and Size for Volatility-Conditioned Trading}\label{locale}

For \(\rho=1\), the \(\{\Pi^\mathbf{1}_t=v\}\)- and \(\{\Pi^\mathbf{1}_t=\underline{v}\}\)-contributions to (\ref{volx}) are invariant under \(v\leftrightarrow\underline{v}\) switching ((\ref{mox}) and (\ref{rv}-\ref{volmix})), 
meaning peak-reward at peak-risk (\(v=\frac{1}{2}\)) in the absence of bias
.

For \(\rho\gg1\), focusing first on the \(\{\Pi^\mathbf{1}_t=v\}\)-contribution to mean-excess (\ref{volx}), it has, by (\ref{momix}) and condition (\ref{window}), a unique optimum (\ref{lowrisk}) to \(\mathcal{O}(\frac{K-1
}{2}t_{\mathbf{p}/t}
)
\)
. The \(\{\Pi^\mathbf{1}_t=\underline{v}\}\)-contributions 
are no more than \(\mathcal{O}(\rho^{-t_{\rho/t}})\), by (\ref{rv}) and condition (\ref{window}); so the leading-order solution (\ref{lowrisk}) is unaffected by them under \( t_{\rho/t}\gg1\) (bias domination of data)\footnote{\label{uneven'}When the unconditional-mix is uneven: for \(\mathbb{P}_t(+)<\frac{1}{2}\) (\(>\frac{1}{2}\)), peak-locations move left (right) \emph{vs} the even-case solution, and peak-size, down (up). All qualitative features remain unchanged. 
    }.
\end{appendix}


\end{document}